\documentclass{article}
\usepackage{spconf,amsmath,amssymb,amsfonts,graphicx}
\usepackage{cite}
\usepackage[ruled,norelsize]{algorithm2e}
\SetKwBlock{Repeat}{repeat}{}
\SetKwFor{At}{at}{do}{end}
\usepackage{array}
\usepackage[font=small,labelfont=bf]{caption}
\usepackage{subcaption}
\usepackage{bbm}
\usepackage{comment}
\usepackage{hhline}
\usepackage[flushleft]{threeparttable}
\usepackage{nccmath}
\usepackage{multicol}
\usepackage{makecell}
\usepackage{tabularx}
\usepackage[hidelinks]{hyperref}
\usepackage{xcolor}
\usepackage{amsthm}
\newtheorem{thm}{Theorem}

\newtheorem{asmp}{Assumption}

\newcolumntype{Y}{>{\centering\arraybackslash}X}
\newcolumntype{b}{>{\hsize=1.65\hsize}Y}
\newcolumntype{z}{>{\hsize=0.9\hsize}Y}
\newcolumntype{s}{>{\hsize=.45\hsize}Y}

\makeatletter
\newcommand{\removelatexerror}{\let\@latex@error\@gobble}
\makeatother

\title{A Distributed Adaptive Algorithm for Node-Specific Signal Fusion Problems in Wireless Sensor Networks}
\name{Cem Ates Musluoglu and Alexander Bertrand\thanks{This project has received funding from the European Research Council (ERC) under the European Union's Horizon 2020 research and innovation programme (grant agreement No 802895). The authors also acknowledge the financial support of the FWO (Research Foundation Flanders) for project G081722N, and the Flemish Government (AI Research Program).}}
\address{\textit{KU Leuven, Department of Electrical Engineering (ESAT), STADIUS Center for Dynamical Systems,}\\ \textit{Signal Processing and Data Analytics, Kasteelpark Arenberg 10, box 2446, 3001 Leuven, Belgium} \\
\{cemates.musluoglu, alexander.bertrand\}@esat.kuleuven.be}
\begin{document}
\ninept

\setlength{\abovedisplayskip}{4pt}
\setlength{\belowdisplayskip}{4pt}
\setlength{\textfloatsep}{7pt plus 1.0pt minus 2.0pt}
\maketitle
\begin{abstract}
  Wireless sensor networks consist of sensor nodes that are physically distributed over different locations. Spatial filtering procedures exploit the spatial correlation across these sensor signals to fuse them into a filtered signal satisfying some optimality condition. However, gathering the raw sensor data in a fusion center to solve the problem in a centralized way would lead to high energy and communication costs. The distributed adaptive signal fusion (DASF) framework has been proposed as a generic method to solve these signal fusion problems in a distributed fashion, which reduces the communication and energy costs in the network. The DASF framework assumes that there is a common goal across the nodes, i.e., the optimal filter is shared across the network. However, many applications require a node-specific objective, while all these node-specific objectives are still related via a common latent data model. In this work, we propose the DANSF algorithm which builds upon the DASF framework, and extends it to allow for node-specific spatial filtering problems.
\end{abstract}
\begin{keywords}
Distributed Signal Processing, Distributed Spatial Filtering, Feature Fusion.
\end{keywords}
\section{Introduction}
\label{sec:intro}
\looseness=-1
Wireless sensor networks (WSNs) consist of a set of physically distributed wireless sensor nodes that are able to locally process the collected sensor data and share it with other nodes in the network. Typically, the goal is to estimate a signal or parameter satisfying some optimality condition which is dependent on the global data of the network, obtained by combining the data collected at every node. In our work, we are interested in optimal spatial filtering \cite{haykin2010handbook}, i.e., linearly combining all the signals measured within a WSN to obtain a filtered output signal that is optimal in some sense. Applications of spatial filtering include --- but are not restricted to --- wireless communication \cite{bjornson2020scalable,sanguinetti2019toward,bashar2019uplink,nayebi2016performance}, biomedical signal processing \cite{wu2017spatial,bertrand2015distributed,blankertz2007optimizing} and acoustics \cite{furnon2021distributed,zhang2018rate,benesty2008microphone}. 

Some applications require the nodes to estimate a common spatial filter as in \cite{lopes2007incremental,lopes2008diffusion,markovich2012distributed}. This usually translates mathematically as an optimization problem which is common to every node in the network. However, a node-specific spatial filter can be desired, e.g., when each node is interested in different source signals or differently filtered versions of the of the same source signal(s) \cite{furnon2021distributed,markovich2009multichannel,doclo2006theoretical}. Each node then has a different optimization problem to solve, i.e., the problem is node-specific, yet can be related, e.g., via a common latent signal model, in which case a joint processing is desirable. 

Distributed algorithms for some particular node-specific problems have been studied, such as minimum mean-squared error (MMSE) \cite{szurley2016topology,bertrand2010danse,doclo2009reduced} and linearly constrained minimum variance beamforming (LCMV) \cite{bertrand2011lcdanse,golan2010reduced,guo2021distributed}, although each problem has been treated separately in the literature. Other distributed algorithms for node-specific problems have been proposed in \cite{chen2015diffusion,plata2015distributed,nassif2017diffusion}, but can generally not be applied to spatial filtering due to the way the data is partitioned across the network. 

The distributed adaptive signal fusion (DASF) algorithm proposed in \cite{musluoglu2022unifiedp1} is a generic ``vanilla'' algorithm that allows to solve spatial filtering problems in a distributed and adaptive fashion, i.e., without centralization of the data, while converging to the solution of the central problem under mild constraints \cite{musluoglu2022unifiedp2}. The DASF algorithm captures several existing distributed signal fusion algorithms as special cases. However, it considers a single common optimization problem to be solved across the network, i.e., it does not allow for node-specific optimization problems. In this work, we propose the distributed adaptive node-specific signal fusion (DANSF) algorithm, which builds upon the DASF framework to solve generic node-specific problems in a distributed fashion. The DANSF algorithm converges to the optimal solution of each node-specific problem under the same assumptions as the original DASF algorithm.

\section{Problem Setting}
\label{sec:dasf_review}
We consider a sensor network with $K$ nodes given in the set $\mathcal{K}=\{1,\dots,K\}$ and connected following a topology given by a graph $\mathcal{G}$, where each link between two nodes $k$ and $l$ implies that nodes $k$ and $l$ can share data with each other. Every node senses an $M_k$-channel signal $\mathbf{y}_k$ so that the network-wide signal can be defined as
\begin{equation}\label{eq:y_part}
  \mathbf{y}=[\mathbf{y}_1^T,\dots,\mathbf{y}_K^T]^T,
\end{equation}
while an observation at time sample $t$ is denoted as $\mathbf{y}(t)\in\mathbb{R}^{M}$, where $M=\sum_{k}M_k$. The signal $\mathbf{y}$ should be viewed as a multivariate stochastic variable, assumed to be ergodic and (short-term) stationary. Each node $k$ acts as a data sink and is interested in finding its own optimal (network-wide) spatial filter $X_k\in\mathbb{R}^{M\times Q}$ and the corresponding filter output $X_k^T\mathbf{y}(t)$, which should satisfy a node-specific optimality condition. We envisage a generic problem statement where we assume that the optimal filter $X_k$ is the solution of an optimization problem of the following form (for node $k$):
\begin{equation}\label{eq:prob}
  \begin{aligned}
  \mathbb{P}_k:\;\underset{X_k\in\mathbb{R}^{M\times Q}}{\text{minimize } } \quad & \varphi_k\big(X_k^{T}\mathbf{y}(t),X_k^{T}B\big)\\
  \textrm{subject to} \quad & \eta_{k,j}\big(X_k^{T}\mathbf{y}(t),X_k^{T}B\big)\leq 0\;\textrm{ $\forall j\in\mathcal{J}_I$,}\\
    & \eta_{k,j}\big(X_k^{T}\mathbf{y}(t),X_k^{T}B\big)=0\;\textrm{ $\forall j\in\mathcal{J}_E$,}
  \end{aligned}
\end{equation}
where the sets $\mathcal{J}_\mathcal{I}$ and $\mathcal{J}_\mathcal{E}$ represent the index sets for inequality and equality constraints respectively. Some examples of problems of the form (\ref{eq:prob}) are shown in Table \ref{tab:ex_prob}. The functions $\varphi_k$ and $\eta_{k,j}$, $j\in\mathcal{J}_\mathcal{I}\cup\mathcal{J}_\mathcal{E}$ are real and scalar-valued functions, while the subscript $k$ specifies that a function or variable is specific for node $k$. Moreover, as the filter output $(X_k^T\mathbf{y}(t))$ is a stochastic variable, the functions in (\ref{eq:prob}) should contain an operator to extract a real-valued quantity from this term, such as an expectation operator. Note that a solution $X_k^*(t)$ of (\ref{eq:prob}) depends on the time sample $t$, as the statistics of the signal $\mathbf{y}$ are allowed to change in time. The proposed DANSF algorithm will act as a block-adaptive filter that estimates and tracks the changes in the data statistics. However, from short-term stationarity, we assume that a solution $X_k^*(t)$ of (\ref{eq:prob}) changes slowly in time compared to the convergence rate of the DANSF algorithm. Therefore we omit the time-dependence of solutions of (\ref{eq:prob}) for mathematical tractability and assume stationarity within convergence time of the algorithm. 

$B$ is a deterministic $M\times L$ matrix independent of the time index $t$, and is commonly encountered to enforce a structure on the variable $X_k$ (see the LCMV example in Table \ref{tab:ex_prob}). Although treated similarly, the distinction between stochastic signals $\mathbf{y}$ and deterministic matrices $B$ are made to emphasize the adaptive and stochastic properties of the algorithm we will present. Similarly to the partitioning of $\mathbf{y}$ in (\ref{eq:y_part}), these deterministic matrices are assumed to be obtained by stacking $M_k\times L$ matrices $B_k$, where $B_k$ is supposed to be available at node $k$, i.e., $B=[B_1^T,\dots,B_K^T]^T$. Moreover, for a fixed node $k$, we also allow Problem (\ref{eq:prob}) to have multiple variables, signals and deterministic matrices (in addition to $X_k$, $\mathbf{y}$ and $B$, respectively), which are however not represented in (\ref{eq:prob}) for conciseness. We assume that every parameter in (\ref{eq:prob}), except $\mathbf{y}$ and $B$, is available at node $k$. For example, the signal $\mathbf{d}_k$ and parameter $H_k$ in the MMSE and LCMV examples respectively of Table \ref{tab:ex_prob} should be available at node $k$. 

For each node $k$, let us denote by $\mathcal{X}_k^*$ the solution set of $\mathbb{P}_k$ and $X_k^*\in\mathcal{X}_k^*$ a specific solution. We then make the following assumption on the set of problems $\mathbb{P}_k$ which links the solutions across the nodes.
\begin{asmp}\label{asmp:proportional}
  There exists a set of invertible $Q\times Q$ matrices $\{D_{k,l}\}_{(k,l)\in\mathcal{K}^2}$ such that for any pair $(k,l)$ of nodes, the solutions $X_k^*\in\mathcal{X}_k^*$ and $X_l^*\in\mathcal{X}_l^*$ satisfy $X_k^*=X_l^*\cdot D_{k,l}$.
\end{asmp}

\begin{table}
  \renewcommand{\arraystretch}{1.5}
  \caption{Examples of problems with node-specific objectives as in (\ref{eq:prob}). MMSE is the minimum mean squared error problem and LCMV the linearly constrained minimum variance beamforming problem. $\mathbb{E}$ denotes the expectation operator.}
  \label{tab:ex_prob}
  \begin{tabularx}{\columnwidth}{|s|b|z|}
  \hline
   Problem & Cost function $\varphi_k$ & Constraints \\ \hhline{|=|=|=|}
   MMSE & $\min \mathbb{E}[||\mathbf{d}_k(t)-X_k^{T}\mathbf{y}(t)||^2]$ & --- \\ \hline
   LCMV  & $\min \mathbb{E}[||X_k^{T}\mathbf{y}(t)||^2]$ & $X_k^{T}B=H_k$ \\ \hline
  \end{tabularx}
\end{table}

These properties were also exploited in \cite{bertrand2010danse,szurley2016topology,bertrand2011lcdanse,guo2021distributed,golan2010reduced,doclo2009reduced} for the design of distributed fusion algorithms for MMSE and LCMV problems. Taking the example of the MMSE problem in Table \ref{tab:ex_prob}, we have $X_k^*=R_{\mathbf{yy}}^{-1}R_{\mathbf{yd}_k}$, with $R_{\mathbf{yy}}=\mathbb{E}[\mathbf{y}(t)\mathbf{y}^T(t)]$ and $R_{\mathbf{yd}_k}=\mathbb{E}[\mathbf{y}(t)\mathbf{d}_k^{T}(t)]$. We then observe that Assumption \ref{asmp:proportional} is satisfied if $\mathbf{d}_k(t)=D_{k,l}^{T}\mathbf{d}_l(t)$. This is true if the desired signals at the different nodes are all different mixtures from the same set of latent sources. This is common in, e.g., hearing aids where the acoustic mixing process needs to be preserved for spatial hearing \cite{golan2010reduced,doclo2009reduced}. For the LCMV example of Table \ref{tab:ex_prob}, a similar argument shows that Assumption \ref{asmp:proportional} is satisfied if $H_k=D_{k,l}^{T}H_l$. In this paper, we propose a unifying algorithmic framework, which has \cite{bertrand2011lcdanse,bertrand2010danse,golan2010reduced,doclo2009reduced} as special cases, while also admitting new problems (see, e.g., Section \ref{sec:sim}), assuming they can be written in the form (\ref{eq:prob}).

\section{DASF for Node-Specific Problems}

In this section, we derive the DANSF algorithm which extends the DASF framework \cite{musluoglu2022unifiedp1,musluoglu2022unifiedp2} to also admit node-specific problems of the form (\ref{eq:prob}). We refer to \cite{musluoglu2022unifiedp1} for a thorough presentation of the DASF algorithm, from which we here only extract the essential ingredients to allow us to define the proposed DANSF algorithm.

At each iteration $i$, an updating node $q\in\mathcal{K}$ is selected and the network represented by the graph $\mathcal{G}$ is temporarily pruned to a tree $\mathcal{T}^i(\mathcal{G},q)$, such that there is a unique path between any pair of nodes in the network. The pruning function is a free design choice, however it should not remove any links between node $q$ and its neighbors $n\in\mathcal{N}_q$ \cite{musluoglu2022unifiedp1}, where $\mathcal{N}_q$ denotes the set of neighboring nodes of node $q$. In the remaining parts of the algorithm derivation, the neighbors of any node are defined to be the ones after pruning the network, i.e., based on the edges of $\mathcal{T}^i(\mathcal{G},q)$. Note that the updating node changes at each iteration $i$ of the algorithm, which implies that a different tree is used in each iteration.

Let us partition each $X_k$, i.e., the network-wide spatial filter that generates the desired node-specific output signal for node $k$, as
\begin{equation}\label{eq:X_part}
  X_k=[X_{k1}^{T},\dots,X_{kK}^{T}]^T,
\end{equation}
such that each $X_{kl}$ is $M_l\times Q$. For each $k\in\mathcal{K}$, we define $X_{kk}$, the $k-$th block of $X_k$, to be the compressor at node $k$. At iteration $i$, every node $k\neq q$ uses its current estimate of $X_{kk}$ to compress the local $M_k-$dimensional sensor signal $\mathbf{y}_k$ into a $Q-$dimensional one. A similar compression applied to node $k$'s matrix $B_k$ leads to
\begin{equation}\label{eq:y_B_compress}
  \widehat{\mathbf{y}}_k^i\triangleq X_{kk}^{iT}\mathbf{y}_k,\;\widehat{B}_k^i\triangleq X_{kk}^{iT}B_k,
\end{equation}
where $X_{kk}^0$ is initialized randomly for each $k$. The nodes will then fuse and forward their compressed data (\ref{eq:y_B_compress}) towards node $q$ as explained next. We note that a batch of $N$ time samples of $\widehat{\mathbf{y}}_k^i$ should be transmitted by node $k$, where $N$ should be chosen such that there are enough samples to estimate the relevant statistics of $\widehat{\mathbf{y}}_k^i$ that are used in the objective and constraints of (\ref{eq:prob}) (in most practical examples, this consists of all the second-order statistics). At each iteration, a different block of $N$ samples is used, so that in the case of changes in statistics of the signal $\mathbf{y}$, the proposed method can adaptively track these changes. Each node $k$ first waits until receiving data from all of its neighbors except one, which we denote as $n$. Node $k$ then fuses the data received from its neighbors $l\in\mathcal{N}_k\backslash\{n\}$ with its data (\ref{eq:y_B_compress}) and the result is then transmitted to node $n$, which receives $N$ samples of
\begin{equation}\label{eq:sum_fwd}
  \widehat{\mathbf{y}}_{k \rightarrow n}^i\triangleq X_{kk}^{iT}\mathbf{y}_k+\sum_{l\in\mathcal{N}_k\backslash\{n\}}\widehat{\mathbf{y}}_{l\rightarrow k}^i,
\end{equation}
where $\widehat{\mathbf{y}}_{l\rightarrow k}^i$ is the data received by node $k$ from its neighbor $l$. Note that the second term of (\ref{eq:sum_fwd}) is recursive and vanishes for nodes with a single neighbor, i.e., leaf nodes, implying that the recursion in (\ref{eq:sum_fwd}) is bootstrapped at the leaf nodes of the tree $\mathcal{T}^{i}(\mathcal{G},q)$. The fused data eventually reaches node $q$, which receives $N$ samples of
\begin{equation}\label{eq:sum_fwd_n}
  \widehat{\mathbf{y}}_{n\rightarrow q}^i=X_{nn}^{iT}\mathbf{y}_n+\sum_{k\in\mathcal{N}_n\backslash\{q\}}\widehat{\mathbf{y}}_{k\rightarrow n}^i=\sum_{k\in\mathcal{B}_{nq}}\widehat{\mathbf{y}}_k^i,
\end{equation}
from all its neighbors $n\in\mathcal{N}_q$. In (\ref{eq:sum_fwd_n}), $\mathcal{B}_{nq}$ is defined to be the subgraph of $\mathcal{T}^i(\mathcal{G},q)$ which contains node $n$ and obtained after cutting the link between nodes $n$ and $q$. A similar recursion applies for the compressed matrices $\widehat{B}_k^i$, and we define $\widehat{B}_{n\rightarrow q}^i$ as the matrix analogous to (\ref{eq:sum_fwd_n}) received by node $q$. Defining $\mathcal{N}_q\triangleq\{n_1,\dots,n_{|\mathcal{N}_q|}\}$, the data collected at node $q$ can be structured as
\begin{equation}\label{eq:tree_data}
  \begin{aligned}
    \widetilde{\mathbf{y}}_q^i&\triangleq[\mathbf{y}_q^T,\widehat{\mathbf{y}}_{n_1\rightarrow q}^{iT},\dots,\widehat{\mathbf{y}}_{n_{|\mathcal{N}_q|}\rightarrow q}^{iT}]^T\;\in\mathbb{R}^{\widetilde{M}_q},\\
    \widetilde{B}_q^i&\triangleq [B_q^T,\widehat{B}_{n_1\rightarrow q}^{iT},\dots,\widehat{B}_{n_{|\mathcal{N}_q|}\rightarrow q}^{iT}]^T\;\in\mathbb{R}^{\widetilde{M}_q\times L}
  \end{aligned}
\end{equation}
where $\widetilde{M}_q=|\mathcal{N}_q|\cdot Q+M_q$.
Using the local data in (\ref{eq:tree_data}), node $q$ creates a compressed version of its original problem $\mathbb{P}_q$ given by
\begin{equation}\label{eq:loc_prob}
  \begin{aligned}
  \underset{\widetilde{X}_q\in\mathbb{R}^{\widetilde{M}_q\times Q}}{\text{minimize } } \quad & \varphi_q\big(\widetilde{X}_q^T\widetilde{\mathbf{y}}_q^i(t),\widetilde{X}_q^T\widetilde{B}_q^i\big)\\
  \textrm{subject to} \quad & \eta_{q,j}\big(\widetilde{X}_q^T\widetilde{\mathbf{y}}_q^i(t),\widetilde{X}_q^T\widetilde{B}_q^i\big)\leq 0\;\textrm{ $\forall j\in\mathcal{J}_I$,}\\
    & \eta_{q,j}\big(\widetilde{X}_q^T\widetilde{\mathbf{y}}_q^i(t),\widetilde{X}_q^T\widetilde{B}_q^i\big)=0\;\textrm{ $\forall j\in\mathcal{J}_E$.}
  \end{aligned}
\end{equation}
Note that (\ref{eq:loc_prob}) has the same objective and constraint functions as (\ref{eq:prob}) (for $k=q$), hence a solver for (\ref{eq:prob}) can also be used locally at node $q$ to solve the compressed problem (\ref{eq:loc_prob}). This is an interesting feature of the DASF framework, which is inherited in DANSF as well.

Node $q$ then solves its local problem (\ref{eq:loc_prob}) to obtain $\widetilde{X}_q^{i+1}$. In the cases where (\ref{eq:loc_prob}) has multiple solutions, we choose $\widetilde{X}_q^{i+1}$ by minimizing $||\widetilde{X}_q-\widetilde{X}_q^{i}||_F$ over all solutions of Problem (\ref{eq:loc_prob}), with
\begin{equation}\label{eq:X_fixed}
  \widetilde{X}_q^i=[X_{qq}^{iT},I_Q,\dots,I_Q]^T
\end{equation}
and where $X_{qq}^{i}$ corresponds to the current estimate of the compressor $X_{qq}$ of node $q$.
The optimal solution $\widetilde{X}_q^{i+1}$ is then partitioned as
\begin{equation}\label{eq:X_tilde_sol}
  \widetilde{X}_q^{i+1}=[X_{qq}^{(i+1)T},G_{qn_1}^{(i+1)T},\dots,G_{qn_{|\mathcal{N}_q|}}^{(i+1)T}]^T,
\end{equation}
with each $G-$matrix being $Q\times Q$. The new estimate of the variable $X_q$ at iteration $i$ is then
\begin{equation}\label{eq:X_q_est}
  X_{qk}^{i+1}=\begin{cases}
  X_{qq}^{i+1} & \text{if $k=q$} \\
  X_{kk}^{i}G_{qn}^{i+1} & \text{if $k\in\mathcal{B}_{nq}$, $n\in\mathcal{N}_q$}.
  \end{cases}
\end{equation}
For nodes $k\neq q$, a new estimate of their variable $X_k$ can be estimated in the following way. Node $q$ first transmits
\begin{equation}\label{eq:full_compress}
  \mathbf{z}_q^{i+1}(t)\triangleq\widetilde{X}_q^{(i+1)T}\widetilde{\mathbf{y}}_q^i(t),\quad Z_q^{i+1}\triangleq\widetilde{X}_q^{(i+1)T}\widetilde{B}_q^i
\end{equation}
to each node $k$, which can either be broadcast by node $q$ or transmitted following the pruned network topology (see \cite{szurley2016topology} for a discussion on efficient ways to achieve this). Note that again $N$ samples of $\mathbf{z}_q^{i+1}$ should be sent by node $q$ to the other nodes. Node $k$ then solves
\begin{equation}\label{eq:loc_prob_k}
  \begin{aligned}
  \underset{F_{kq}\in\mathbb{R}^{Q\times Q}}{\text{minimize } } \quad & \varphi_k\big(F_{kq}^T\mathbf{z}_q^{i+1}(t),F_{kq}^TZ_q^{i+1}\big)\\
  \textrm{subject to} \quad & \eta_{k,j}\big(F_{kq}^T\mathbf{z}_q^{i+1}(t),F_{kq}^TZ_q^{i+1}\big)\leq 0\;\textrm{ $\forall j\in\mathcal{J}_I$,}\\
    & \eta_{k,j}\big(F_{kq}^T\mathbf{z}_q^{i+1}(t),F_{kq}^TZ_q^{i+1}\big)=0\;\textrm{ $\forall j\in\mathcal{J}_E$,}
  \end{aligned}
\end{equation}
such that a new estimate of its variable $X_k$ at iteration $i$ is given by
\begin{equation}\label{eq:update_X_k}
  X_k^{i+1}=X_q^{i+1}F_{kq}^{i+1},
\end{equation}
where $F_{kq}^{i+1}$ is a solution of (\ref{eq:loc_prob_k}) at node $k$. Note that the compressor at node $k$, i.e., $X_{kk}$ is part of $X_k$, i.e., the compression matrix of node $k$ is also updated by (\ref{eq:update_X_k}). Since only the compressor matrices $X_{kk}$ (for all $k$) play a role within the algorithm (as these define the transmitted signals), the update of the blocks $X_{kl}$ with $k\neq l$ can be omitted, unless the nodes are explicitly interested in knowing the coefficients of the full matrix $X_k$. However, in most applications, the filter output signal $\mathbf{z}_k(t)=X_k^T\mathbf{y}(t)$ is sought after, rather than the filter $X_k$ itself, which can be computed at each node $k$ as
\begin{equation}\label{eq:est_filter}
  \mathbf{z}_k^{i+1}(t)\triangleq\begin{cases}
    \widetilde{X}_q^{(i+1)T}\widetilde{\mathbf{y}}_q^i(t) & \text{if $k=q$} \\
    F_{kq}^{(i+1)T}\mathbf{z}_q^{i+1}(t) & \text{if $k\neq q$},
    \end{cases}
\end{equation}
without keeping track of other subblocks $X_{kq}$ for $k\neq q$. This is because the filtering of subblocks $X_{kq}$ is done at node $q$ instead of node $k$, using the compressor $X_{qq}$, of which the output is transformed with $F_{kq}$ at node $k$.

\looseness=-1
At each iteration, this process is repeated by selecting a different updating node. Algorithm \ref{alg:dasf_k} summarizes the steps of the DANSF algorithm described above. We note that \cite{bertrand2011lcdanse,bertrand2010danse,golan2010reduced,doclo2009reduced} are special cases of this algorithm. The method is able to adapt to and track changes in the signal statistics of $\mathbf{y}$, as is the case for the original DASF algorithm. This is because a new block of $N$ samples, e.g., $\{\mathbf{y}(t)\}_{t=iN}^{(i+1)N-1}$, is measured and used at each iteration to solve (\ref{eq:loc_prob}), i.e., different iterations of the DANSF algorithm are spread over different sample blocks across the time dimension, similar to an adaptive filter, making each $X_k^i$ an estimate of $X_k^*(t)$, where $t=iN$.

The following theorem guarantees convergence in cost for the DANSF algorithm:

\begin{figure}[!t]
  \removelatexerror
  \DontPrintSemicolon
  \begin{algorithm}[H]
  \caption{DANSF algorithm}\label{alg:dasf_k}
  \SetKwInOut{Output}{output}
  \BlankLine
  $X_{kk}^0$ initialized randomly for each $k$, $i\gets0$.\;
  \Repeat
  {
  Choose the updating node as $q\gets (i\mod K)+1$.\;
  1) The network $\mathcal{G}$ is pruned into a tree $\mathcal{T}^i(\mathcal{G},q)$.\;

  2) Each node $k$ collects $N$ samples of $\mathbf{y}_k$, compress these to $N$ samples of $\widehat{\mathbf{y}}^{i}_k$ and also compute $\widehat{B}_k^i$ as in (\ref{eq:y_B_compress}).\;

  3) The nodes sum-and-forward their compressed data towards node $q$ via the recursive rule (\ref{eq:sum_fwd}) (and a similar rule for the $\widehat{B}_k^i$'s). Node $q$ eventually receives $N$ samples of $\widehat{\mathbf{y}}^{i}_{n\rightarrow q}$ given in (\ref{eq:sum_fwd_n}) and similarly, $\widehat{B}_{n\rightarrow q}^i$, from all its neighbors $n\in\mathcal{N}_q$.\; 

  \At{Node $q$}
  {
    4a) Compute the solution of the local problem (\ref{eq:loc_prob}) to obtain $\widetilde{X}_q^{i+1}$. If the solution is not unique, select $\widetilde{X}_q^{i+1}$ which minimizes $||\widetilde{X}_q^{i+1}-\widetilde{X}_q^i||_F$ with $\widetilde{X}_q^i$ defined in (\ref{eq:X_fixed}).\;
    4b) Partition $\widetilde{X}^{i+1}_q$ as in (\ref{eq:X_tilde_sol}).\;
    4c) Update the estimate of $X_q$ as in (\ref{eq:X_q_est}).\;
    4d) Disseminate $Z_q^{i+1}$ and $N$ samples of $\mathbf{z}_q^{i+1}$ as defined in (\ref{eq:full_compress}) within the tree to each data sink node.
  }
  
  5) Nodes $k\neq q$ update $X_k$ according to (\ref{eq:update_X_k}) by solving (\ref{eq:loc_prob_k}) and can estimate its filtered output as in (\ref{eq:est_filter}).\;
  
  $i\gets i+1$\;
  }
  \end{algorithm}
\end{figure}
 
\begin{thm}\label{thm:f_conv}
  \looseness=-1
  Let us denote by $(\varphi_k^i)_i$ the sequence of function values $\varphi_k\left(X_k^{iT}\mathbf{y}(t),X_k^{iT}B\right)$, $\forall k\in\mathcal{K}$, obtained from Algorithm \ref{alg:dasf_k}. Then, the sequence $(\varphi_k^i)_i$ is non-increasing, and converges for each $k$, i.e., the cost function at each node is monotonically decreasing. Furthermore, $X_k^i$ always satisfies the constraints of (\ref{eq:prob}) for each $k$ and $i>0$.
\end{thm}

\begin{proof}
  For conciseness, we omit the matrix $B$ in this proof (it can be treated similarly to $\mathbf{y}$ within the proof). We observe from (\ref{eq:sum_fwd_n})-(\ref{eq:tree_data}) that there exists a linear relationship between $\widetilde{\mathbf{y}}_q^i$ and $\mathbf{y}$ such that there exists a compressive $M\times \widetilde{M}_q$ matrix $C_q^i$ for which $\widetilde{\mathbf{y}}_q^i(t)=C_q^{iT}\mathbf{y}(t)$. Using this matrix $C_q^i$, we can write $\widetilde{X}_q^{T}\widetilde{\mathbf{y}}_q^i(t)=\widetilde{X}_q^T(C_q^{iT}\mathbf{y}(t))=(C_q^i\widetilde{X}_q)^T\mathbf{y}(t)$. Therefore, at iteration $i$ and updating node $q$, we can observe that $X_q$ is parameterized as $X_q=C_q^i\widetilde{X}_q$. Substituting $\widetilde{X}_q^T\widetilde{\mathbf{y}}_q^i(t)$ with $(C_q^i\widetilde{X}_q)^T\mathbf{y}(t)$ within (\ref{eq:loc_prob}), we find that any point $C_q^i\widetilde{X}_q$ is a feasible point of the problem (\ref{eq:prob}) at node $q$ if $\widetilde{X}_q$ is a feasible point of the local problem (\ref{eq:loc_prob}) at node $q$, and vice versa. In particular, since $X_q^{iT}\mathbf{y}(t)=\widetilde{X}_q^{iT}\widetilde{\mathbf{y}}_q^i(t)$ for $\widetilde{X}_q^i$ given in (\ref{eq:X_fixed}), we have $X_q^i=C_q^i\widetilde{X}_q^i$. Additionally, from (\ref{eq:X_tilde_sol})-(\ref{eq:X_q_est}), it can be shown that $X_q^{i+1}=C_q^i\widetilde{X}_q^{i+1}$. Therefore, both $\widetilde{X}_q^{i}$ and $\widetilde{X}_q^{i+1}$ are feasible points of the local problem (\ref{eq:loc_prob}), i.e., belong to $\widetilde{\mathcal{S}}_q^i$, which we define to be the constraint set of (\ref{eq:loc_prob}). Let us now define the functions $f_k$ as 
  \begin{equation}\label{eq:phi_compact}
    f_k(X_k)=\varphi_k\big(X_k^T\mathbf{y}(t)\big),\;\forall k\in\mathcal{K}.
  \end{equation}
  For the updating node $q$, at iteration $i$, we have that $f_q(C_q^i\widetilde{X}_q^{i+1})\leq f_q(C_q^i\widetilde{X}_q)$ for any $\widetilde{X}_q\in\widetilde{\mathcal{S}}_q^i$, since $\widetilde{X}_q^{i+1}$ solves the local problem (\ref{eq:loc_prob}). Moreover, since we have shown above that $\widetilde{X}_q^i\in\widetilde{\mathcal{S}}_q^i$, we have $f_q(C_q^i\widetilde{X}_q^{i+1})=f_q(X_q^{i+1})\leq f_q(C_q^i\widetilde{X}_q^i)=f_q(X_q^i)$. This shows a monotonic decrease of the cost at node $q$. For the case $k\neq q$, let us (hypothetically) assume that the updating node $q$ would have used the cost function of node $k$, i.e., it solves
  \begin{equation}\label{eq:loc_prob_k_q}
    \begin{aligned}
    \underset{\widetilde{X}_q\in\mathbb{R}^{\widetilde{M}_q\times Q}}{\text{minimize } } \quad & \varphi_k\big(\widetilde{X}_q^T\widetilde{\mathbf{y}}_q^i(t)\big)\\
    \textrm{subject to} \quad & \eta_{k,j}\big(\widetilde{X}_q^T\widetilde{\mathbf{y}}_q^i(t)\big)\leq 0\;\textrm{ $\forall j\in\mathcal{J}_I$,}\\
      & \eta_{k,j}\big(\widetilde{X}_q^T\widetilde{\mathbf{y}}_q^i(t)\big)=0\;\textrm{ $\forall j\in\mathcal{J}_E$,}
    \end{aligned}
  \end{equation}
  instead of (\ref{eq:loc_prob}). As mentioned earlier, (\ref{eq:loc_prob}) and (\ref{eq:loc_prob_k_q}) are compressed versions of the problem (\ref{eq:prob}) for node $q$ and node $k$ respectively, i.e., the data $\mathbf{y}$ of (\ref{eq:prob}) is replaced by $\widetilde{\mathbf{y}}_q^i$ in (\ref{eq:loc_prob}) and (\ref{eq:loc_prob_k_q}). Therefore, the local problems (\ref{eq:loc_prob})-(\ref{eq:loc_prob_k_q}) satisfy Assumption \ref{asmp:proportional}, i.e., there exists a matrix $\widetilde{D}_{q,k}$ such that $\widetilde{X}_q^{i+1}=\widetilde{X}_k^{i+1}\cdot \widetilde{D}_{q,k}$. This implies that node $q$ also optimizes $\varphi_k$ up to a transformation with a $Q\times Q$ matrix. The latter transformation is compensated for by finding a proper transformation $F_{kq}$ at node $k$ by solving (\ref{eq:loc_prob_k}). Since this argument holds for any iteration $i$, and from the relationship $X_k^{i+1}=X_q^{i+1}F_{kq}^{i+1}$ in (\ref{eq:update_X_k}), we have $f_k(X_k^{i+1})\leq f_k(X_k^i)$ even though node $q$ optimizes $f_q$ instead of $f_k$ at iteration $i$. From the definition of $f_k$ in (\ref{eq:phi_compact}), $(\varphi_k^i)_i$ is therefore a non-increasing sequence for each $k$. Since these sequences are respectively lower bounded by the minimal value of $\varphi_k$ achieved for $X_k^*$, say $\varphi_k^*$, over the constraint set of $\mathbb{P}_k$ in (\ref{eq:prob}), they are converging sequences.
\end{proof}

The convergence of the DANSF algorithm for each node $k$ to a solution $X_k^{*}\in\mathcal{X}_k^{*}$ of $\mathbb{P}_k$ is summarized in the theorem below. The convergence guarantee is under the same mild technical conditions as the DASF algorithm \cite{musluoglu2022unifiedp1,musluoglu2022unifiedp2}. The proof is omitted due to space constraints, but follows similar steps as in \cite{musluoglu2022unifiedp2}.

\begin{thm}[Proof Omitted]\label{thm:X_conv_thm}
  Suppose that, for each node $k$, Problem (\ref{eq:prob}) satisfies Assumption \ref{asmp:proportional} and the conditions for convergence of the original DASF algorithm\footnote{Due to space constraints, and since some of these conditions are quite technical, we refer the reader to \cite{musluoglu2022unifiedp1,musluoglu2022unifiedp2}.} (see \cite{musluoglu2022unifiedp1,musluoglu2022unifiedp2}). Then the sequences $(X_k^i)_i$, for each $k\in\mathcal{K}$, obtained from the DANSF algorithm also converge respectively to an optimal point $X_k^*\in\mathcal{X}_k^*$ of Problem (\ref{eq:prob}).
\end{thm}

Note that this also implies that each node has access to its optimal node-specific filter output $\mathbf{z}_k^*(t)=X_k^{*T}\mathbf{y}(t)$ for all samples collected after convergence of the algorithm. Sensor observations used \textit{during} convergence of the algorithm are fused suboptimally (similar to how an adaptive filter initially produces suboptimal filter outputs). The algorithm can be used in an adaptive or tracking context if the dynamics of the statistics change slowly, i.e., slower than the convergence time of the DANSF algorithm.

\section{Simulations}
\label{sec:sim}

We consider the following node-specific problem at each node $k$
\begin{equation}\label{eq:qclp_k}
  \begin{aligned}
  \underset{X_k\in\mathbb{R}^{M\times Q}}{\text{minimize } } & \quad \text{trace}(X_k^{T}B_k)\\
  \textrm{subject to } & \quad \mathbb{E}[||X_k^{T}\mathbf{y}(t)||^2]=\text{trace}(X_k^{T}R_{\mathbf{yy}}X_k)\leq 1,
  \end{aligned}
\end{equation}
where $R_{\mathbf{yy}}=\mathbb{E}[\mathbf{y}(t)\mathbf{y}^T(t)]$. Taking $B_k\neq 0$, the unique solution of the problem is given by $X_k^{*}=-\beta_k\cdot R_{\mathbf{yy}}^{-1}B_k$, 
where $\beta_k=\sqrt{\text{trace}(B_k^{T}R_{\mathbf{yy}}^{-1}B_k)^{-1}}$. For each $k$, we take $B_k=B\cdot D_k$, where each element of $B$ and $D_k$'s are drawn from a Gaussian distribution with zero-mean and variance $1$, i.e., $\mathcal{N}(0,1)$. Note that Problem (\ref{eq:qclp_k}) satisfies Assumption \ref{asmp:proportional} as $X_k^{*}=X_l^{*}\cdot D_{k,l}$, where $D_{k,l}=D_l^{-1}D_k\cdot\beta_k/\beta_l$. The signal $\mathbf{y}$ follows the mixture model $\mathbf{y}(t)=A\cdot \mathbf{d}(t)+\mathbf{n}(t)$, where each element of $\mathbf{d}\in\mathbb{R}^Q$ and $\mathbf{n}\in\mathbb{R}^M$ independently follows $\mathcal{N}(0,0.5)$ and $\mathcal{N}(0,0.1)$ respectively, at each time sample. Every entry of $A\in\mathbb{R}^{M\times Q}$ is drawn from $\mathcal{N}(0,0.2)$. At each iteration $i$ of the DANSF algorithm, the updating node $q$ solves its local problem (\ref{eq:loc_prob}) given by
\begin{equation}\label{eq:qclp_k_loc}
    \underset{\widetilde{X}_q\in\mathbb{R}^{\widetilde{M}_q\times Q}}{\text{min. } } \text{trace}(\widetilde{X}_q^{T}\widetilde{B}_q^i),\quad \textrm{s. t. }\text{trace}(\widetilde{X}_q^{T}R^i_{\widetilde{\mathbf{y}}_q\widetilde{\mathbf{y}}_q}\widetilde{X}_q)\leq 1,
\end{equation}
where $\widetilde{\mathbf{y}}_q^i$ and $\widetilde{B}_q^i$ are defined as in (\ref{eq:tree_data}) and $R^i_{\widetilde{\mathbf{y}}_q\widetilde{\mathbf{y}}_q}=\mathbb{E}[\widetilde{\mathbf{y}}_q^i(t)\widetilde{\mathbf{y}}_q^{iT}(t)]$. 

\looseness=-1
The performance of the DANSF algorithm is assessed by computing the relative mean squared error (MSE) $\epsilon_k(X_k^i)=||X_k^i-X_k^*||_F^2\cdot||X_k^*||_F^{-2}$, where $X_k^*$ is the solution of (\ref{eq:qclp_k}) for node $k$. In all our experiments, we take $Q=3$, $K=10$ and $M_k=7$ for every $k\in\mathcal{K}$. Figure \ref{fig:nodes_conv} shows the MSE $\epsilon_k$ for every node $k$ and for different network topologies namely fully-connected networks, networks with line topologies, i.e., each node has two neighbors except two which have a single neighbor, and networks with randomly generated topologies. In every experiment, $\mathcal{T}^i(\cdot,q)$ was taken to be the shortest path pruning function. We observe that the DANSF algorithm converges to $X_k^*$ for every node $k$ of the network, as stated in Theorem \ref{thm:X_conv_thm}, although at different convergence rates for different topologies. Fully-connected networks converge the fastest, while the slowest convergence rate is obtained for networks with a line topology. A similar result was observed for the DASF algorithm, where networks with more connected topologies lead to faster convergence rates \cite{musluoglu2022unifiedp1}. Additionally, we see from Figure \ref{fig:nodes_conv} that each node's estimate of its variable $X_k$ converges to the respective optimal value $X_k^*$ without large deviations in convergence rate between different nodes.

\begin{figure}[t]
  \includegraphics[width=0.48\textwidth]{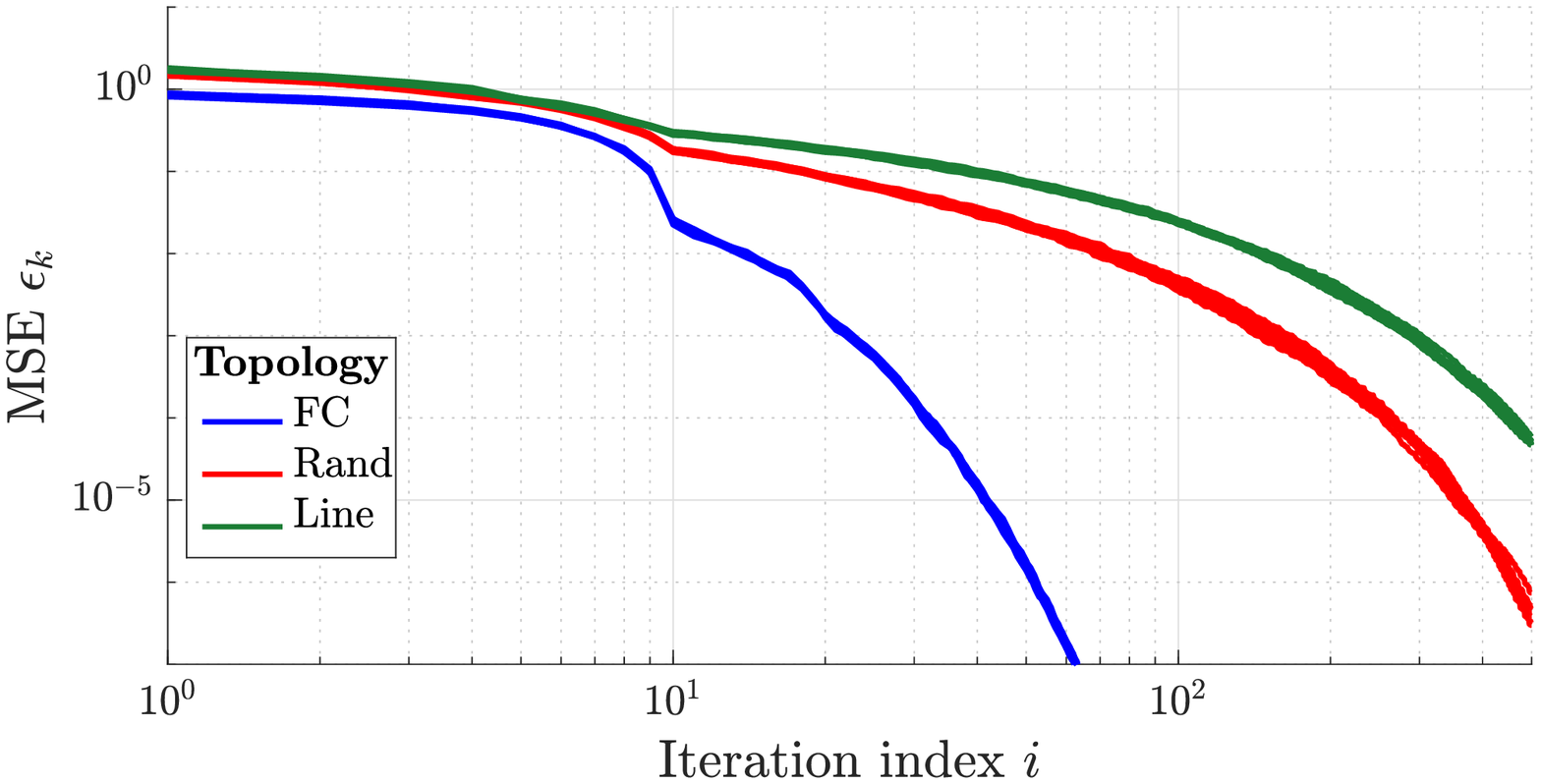}
  \caption{MSE $\epsilon_k$ for all nodes $k$ of the DANSF algorithm for various network topologies, namely fully-connected networks (\textit{FC}), randomly generated networks using the Erd\H{o}s-R\'enyi model (\textit{Rand}) and graphs in a line topology (\textit{Line}). Each point has been obtained by taking the median of 100 Monte-Carlo runs.}
  \label{fig:nodes_conv}
\end{figure}

\section{Conclusion}

We have proposed the DANSF algorithm to solve node-specific signal fusion problems in a distributed fashion over a network. The DANSF algorithm builds upon the principles of the DASF framework and extends it to problems with different optimization problems at each node, yet with coupled solution sets, which leads to analogous convergence results between both algorithms. We provided a proof for the convergence in cost, which showed that we obtain a monotonic decrease of the cost at each node. Simulations of the DANSF algorithm applied on a new problem validated our convergence claims.

\vfill\pagebreak

\bibliographystyle{IEEEbib}
\bibliography{refs}

\end{document}